\documentclass{article}
\usepackage[T1]{fontenc}
\usepackage{graphicx} 

\usepackage[english]{babel}
\usepackage[utf8]{inputenc}
\usepackage{csquotes}
\usepackage{cite}

\usepackage{endnotes}
\usepackage{authblk}

\usepackage{graphicx}
\usepackage[a4paper]{geometry}
\usepackage[pdfborder={0 0 0}]{hyperref}
\usepackage[utf8]{inputenc}
\usepackage{dsfont, amsmath, amsfonts, amssymb, amsthm, units}
\usepackage{enumerate}
\usepackage{tikz-cd}
\usepackage{bbold}
\usepackage{hyperref,bookmark}
\usepackage{slashed}
\usepackage{mathtools}
\usepackage{palatino}
\usepackage{euler}
\usepackage{verbatim}

\DeclarePairedDelimiterX\braket[2]{\langle}{\rangle}{#1\,\delimsize\vert\,\mathopen{}#2}

\newcommand{\R}{\mathbb{R}}

\newcommand{\norm}[1]{\left\lVert#1\right\rVert}

\theoremstyle{definition}
\newtheorem{definition}{Definition}[section]
\newtheorem{example}[definition]{Example}
\newtheorem{remark}[definition]{Remark}
\newtheorem{corollary}[definition]{Corollary}

\newtheorem{lemma}[definition]{Lemma}
\newtheorem{proposition}[definition]{Proposition}

\title{Boundaries in the Instantaneous Formulation of Field Theories} 
\author[1,*]{Silvester G.A. Borsboom}
\affil[1]{\normalsize Institute for Mathematics, Astrophysics and Particle Physics and Radboud Center for Natural Philosophy, Radboud University Nijmegen}
\affil[*]{silvester.borsboom@ru.nl}

\date{\today}

\begin{document}

\maketitle

\begin{abstract}

We study boundary conditions in GiMmsy's covariant and instantaneous formulations of classical field theories and show that the instantaneous state space in the presence of a constant Dirichlet boundary condition is a tangent bundle to the configuration space of fields satisfying said condition. We then study the instantaneous state space when only the velocity of the field is required to vanish at the boundary and show that this results in a sector structure, where each sector is a tangent bundle labeled by the configuration at the boundary. Taking the Legendre transform of this sectored state space yields a sectored phase space with leafwise canonical Poisson structures. We apply this to Yang-Mills theory with spatial boundary conditions and relate our results to flux superselection sectors. The sector-moving gauge transformations are not Hamiltonian because of the lack of a boundary momentum, prompting us to propose a novel definition of the asymptotic or boundary symmetry group as the quotient of the boundary-preserving Hamiltonian transformations by the trivial ones. The physical boundary symmetry group of electromagnetism is then shown to be a copy of the global gauge group even when all sectors are considered simultaneously. Conditions are discussed under which the same holds for non-Abelian Yang-Mills theory.

\end{abstract}

\section{Introduction}\label{intro}

Boundaries play an increasingly prominent role in the study of classical field theories because of their highly non-trivial interaction with localizable symmetry groups \cite{Regge:1974zd}. By the Vanishing Theorem, also known as Noether's second theorem, (time-)localizable symmetries lead to constraints in the Hamiltonian formulation of a field theory \cite{binz,Gotay2006MomentumMA}. Under suitable regularity and well-posedness assumptions, the final constraint set is contained in the zero locus of the instantaneous energy-momentum map. If, in addition, the gauge group is full and all secondary constraints are first class, this inclusion becomes an equality \cite[Theorem 10B.1]{Gotay2006MomentumMA}. The infinitesimal gauge directions lie in the null distribution of the symplectic form pulled back to the final constraint set, meaning that these symmetry transformations are unphysical, and coincide with the unphysical transformations when the gauge group is full.

But in the presence of boundaries, gauge transformations acting non-trivially at the boundary may instead acquire physical significance: on the bulk constraint surface, their momentum may still reduce to a non-vanishing boundary flux. This phenomenon is well studied in the canonical Hamiltonian formulation \cite{rielloHamiltonianGaugeTheory2024,rielloNullHamiltonianYang2025}. However, a well-defined classical field theory needs a full instantaneous formulation, including instantaneous Lagrangian and Hamiltonian formulations, as well as a well-defined instantaneous Legendre transform between them and a 3+1 split from both the Lagrangian and Hamiltonian covariant formulations. 

The covariant and instantaneous formulations of classical field theories have been developed by GiMmsy \cite{gotayMomentumMapsClassical2004,gotayMomentumMapsClassical2004a,Gotay2006MomentumMA}.\footnote{GiMmsy stands for Gotay, Isenberg, Marsden, Montgomery, Sniatycki en Yasskin, with the names of the ``main protagonists'' capitalized \cite{blohmann}.} In this paper we extend the GiMmsy framework to develop a rigorous foundation for the study of boundary conditions in classical field theories, in particular in the instantaneous formulation. This allows us to understand how boundaries interact with 3+1 decompositions, the Legendre transform, Hamiltonian constraints, momentum maps and boundary fluxes.

We begin in Section \ref{sec:instantaneousformulation} with a recap of GiMmsy's framework. In Section \ref{sec:BCs} we then consider various boundary conditions in this framework. This leads us to a study of the instantaneous Legendre transform in the presence of boundary conditions in Section \ref{sec:legendre}. Finally we apply our results to Yang-Mills theory in Section \ref{sec:YM}, extending results in \cite{Borsboom:2025agn}.

\section{The instantaneous formulation of field theories}\label{sec:instantaneousformulation}

In the instantaneous analysis of field theories, developed by GiMmsy \cite{binz,gotayMomentumMapsClassical2004,gotayMomentumMapsClassical2004a}, one moves from the covariant to the instantaneous theory by implementing a 3+1 split of spacetime and defining the associated \textit{instantaneous Lagrangian} and \textit{instantaneous Legendre transform}. We will now summarize this framework.

Let $M$ denote 4-dimensional spacetime and let $\pi_Y\colon Y\to M$ be the \textit{field bundle}, sections of which are the fields of the theory (e.g. the affine connection bundle for Yang-Mills theory on a principal $G$-bundle $P$, which becomes $Y=T^*M\otimes\mathfrak{g}$ when $P$ is trivalized). The covariant configuration space is denoted $Q:=\Gamma(Y)$. A $3+1$ decomposition is performed by first choosing a slicing of $M$ into Cauchy surfaces $\Sigma_t\cong\Sigma$, i.e. a diffeomorphism $\Sigma\times \R\to M$ with generator $\zeta_M\in\mathfrak{X}(M)$ which is transverse to every $\Sigma_t$. A \textit{compatible slicing} of the field bundle $Y$ is a bundle $Y_\Sigma\to \Sigma$ and a diffeomorphism $Y_\Sigma\times \R\to Y$ such that the diagram

$$\begin{array}{ccc}
Y_{\Sigma} \times \mathbb{R} & \longrightarrow & Y \\
\downarrow & & \downarrow \\
\Sigma\times \mathbb{R} & \longrightarrow & M
\end{array}$$
commutes, where the vertical arrows are the bundle projections \cite{gotayMomentumMapsClassical2004a}. The generator of this diffeomorphism is denoted by $\zeta_Y\in \mathfrak{X}(Y)$.

We assume that the Lagrangian depends on first derivatives only. Then a
\textit{covariant Lagrangian density} is a smooth bundle map $\mathfrak{L}\colon J^1Y\to \Lambda^4(M)$, where $ J^1Y$ denotes the first jet bundle of $Y$. This density can be integrated over $M$ to obtain the action:
\begin{align*}
    S(\varphi)=\int_ M\mathfrak{L}(j^1(\varphi)),
\end{align*}
where $j^1\colon\Gamma(Y)\to \Gamma(J^1Y)$ denotes the first jet prolongation. Note that the slicing $\zeta_Y$ induces a slicing of $J^1Y$ through the jet prolongation. Throughout we assume a \textit{Lagrangian} slicing, i.e. such that $\mathfrak{L}$ is equivariant with respect to the
one-parameter groups of automorphisms associated to the induced slicings of $J^1 Y$ and $\Lambda^4(M)$ \cite{gotayMomentumMapsClassical2004a}.

But the action can be written not only as a spacetime integral of a covariant Lagrangian density, but also as a time integral over an instantaneous Lagrangian, in the same way as for standard non-covariant classical mechanics. To obtain the instantaneous Lagrangian one first performs a 3+1 split of the field $\varphi\in\Gamma(Y)$ into
\begin{align*}
    \varphi_t=\varphi|_{\Sigma_t},\;\;\;\;\;\dot{\varphi}_t:=\text{\pounds}_{\zeta}\varphi|_{\Sigma_t}=\left(T\varphi\circ\zeta_M-\zeta_Y\circ\varphi\right)|_{\Sigma_t}.
\end{align*}
We think of $\dot{\varphi}_t$ as the time-derivative of $\varphi$ at time $t$, and it is vertical: $\dot{\varphi}_t\in \Gamma(\varphi_t^*VY_t)$, where $VY_t\to Y_t$ is the vertical bundle $VY_t=\text{ker}(T\pi_t)\subset TY_t$, with $\pi_t=\pi_Y|_{Y_t}\colon Y_t\to\Sigma_t$ the instantaneous field projection. Thus $\varphi_t^*VY_t$ is a bundle over $\Sigma_t$.

This field decomposition yields the \textit{jet decomposition map} $\beta_\zeta\colon (J^1Y)_t\to J^1(Y_t)\times_{Y_t} VY_t$ \cite{gotayMomentumMapsClassical2004a}:
\begin{align*}
    \beta_\zeta(j^1\varphi(x))=(j^1\varphi_t(x),\dot{\varphi}_t(x)).
\end{align*}
Here $(J^1Y)_t$ denotes the restriction of $J^1Y$ over $\Sigma_t$ and $Y_t$ the restriction of $Y$ over $\Sigma_t$. We obtain vertical tangent vectors because the time-derivative yields only variations within the fibers of the field bundle, not within space. 

The jet decomposition map is an affine-bundle isomorphism and its inverse is called the \textit{jet reconstruction map} \cite{gotayMomentumMapsClassical2004a}. Clearly $\beta_\zeta$ can be extended to a map on holonomic sections in $\Gamma((J^1Y)_t)$:
\begin{align*}
    \beta_\zeta(j^1\varphi\circ i_t)=(j^1\varphi_t,\dot{\varphi}_t),
\end{align*}
where $i_t\colon \Sigma_t\to M$ denotes the inclusion of space. The following result explains how the instantaneous state space arises as a tangent bundle \cite[Corollary 6.2]{gotayMomentumMapsClassical2004a}.
\begin{proposition}\label{prop:crucialisom}
    Let $Q_t=\Gamma(Y_t)$ denote the instantaneous configuration space at time $t$. Then $\beta_\zeta$ induces an isomorphism $(j^1Q)_t\cong TQ_t$, where $(j^1Q)_t=\{j^1\varphi\circ i_t:\varphi\in\Gamma(Y)\}$ is the collection of holonomic sections of $J^1Y\to M$ restricted to $\Sigma_t$. 
\end{proposition}

\begin{proof}
This follows from the identification $T_{\varphi_t}Q_t
\cong
\Gamma(\varphi_t^*VY_t)$, given by
\[
v_{\varphi_t}=\dot\gamma(0)\mapsto
\left(x\mapsto \left.\frac{d}{ds}\right|_{s=0}\gamma(s)(x)\right),
\]
where \(\gamma\) is any smooth curve in \(Q_t\) with \(\gamma(0)=\varphi_t\). This is a section of \(\varphi_t^*VY_t\), because for every \(s\) we have $\pi_t(\gamma(s)(x))=x$. The spatial jet is given by \(\varphi_t\in Q_t\), while the velocity (in $\Gamma(\varphi_t^*VY_t)$) determines an element of \(T_{\varphi_t}Q_t\) by this identification. Conversely, jet reconstruction associates a unique jet along \(\Sigma_t\) to this data.
\end{proof}

\begin{remark}
For a holonomic section restricted to $\Sigma_t$, we need no other information for its spatial derivatives than just the section itself on $\Sigma_t$. But its time derivatives must also be included in order to have an isomorphism with $TQ_t$, which is why we must first take holonomic sections and then restrict them to $\Sigma_t$, as this includes the information of their time derivatives at time $t$. In contrast, the jets in $J^1(Y_t)$ only include information about spatial derivatives because $Y_t\to\Sigma_t$ is already restricted to a spatial slice.
\end{remark}

The jet decomposition map can straightforwardly be used to obtain the instantaneous Lagrangian and Legendre transform \cite{gotayMomentumMapsClassical2004a}.

\begin{definition}\label{def:instlagran}
    Given a covariant Lagrangian density $\mathfrak{L}\colon J^1Y\to \Omega^4(M)$ we first define the \textit{instantaneous Lagrangian density} $\mathfrak L_{t,\zeta}\colon
J^1(Y_t)\times_{Y_t}VY_t
\longrightarrow
\Lambda^3\Sigma_t$ by
    \begin{align*}
        \mathfrak L_{t,\zeta}(j^1\varphi_t,\dot\varphi_t)=
i_t^*\iota_{\zeta_M}
\mathfrak L\left(\beta_\zeta^{-1}
(j^1\varphi_t,\dot\varphi_t)\right),
    \end{align*}
    where $\iota_{\zeta_M}$ is just the insertion of $\zeta_M$ into the first slot. The \textit{instantaneous Lagrangian} $\mathcal{L}_{t,\zeta}\in C^\infty(TQ_t)$ is then simply the spatial integral
    \begin{align*}
        \mathcal{L}_{t,\zeta}(\varphi_t,\dot{\varphi}_t)=\int_{\Sigma_t}\mathfrak{L}_{t,\zeta}(j^1\varphi_t,\dot{\varphi}_t),
    \end{align*}
    where $(\varphi_t,\dot{\varphi}_t)$ is understood as an element of $TQ_t$ via the isomorphism from Proposition \ref{prop:crucialisom}.
\end{definition}
The action can then be rewritten as a time-integral of $\mathcal{L}_{t,\zeta}$. We end this Subsection with the following.
\begin{definition}\label{def:instlegendre}
    The \textit{instantaneous Legendre transform} $TQ_t\to T^*Q_t$ is the fiber derivative $\mathbb{F}\mathcal{L}_{t,\zeta}$ defined by
    \begin{align*}
        \mathbb{F}\mathcal{L}_{t,\zeta}(\varphi_t,\dot{\varphi}_t)(\dot{\psi}_{\varphi_t})=\langle d\mathcal{L}_{t,\zeta}(\varphi_t,\dot{\varphi}_t),\dot{\psi}_{\varphi_t}\rangle.
    \end{align*}
    where $\dot{\psi}_{\varphi_t}\in T_{\varphi_t}Q_t$ is viewed as an element of $TTQ_t$ through the canonical identification (i.e. $\mathbb{F}\mathcal{L}_{t,\zeta}$ is only the vertical part of $d\mathcal{L}_{t,\zeta}$, not the full differential) \cite{binz}. Its image is called the \textit{instantaneous primary constraint set}.
\end{definition}
It is assumed that $\text{Im}(\mathbb{F}\mathcal{L}_{t,\zeta})$ is a smooth, closed submanifold of $T^*Q_t$ and that $\mathbb{F}\mathcal{L}_{t,\zeta}$ is a submersion with connected fibers \cite{gotayMomentumMapsClassical2004a}. It can be written more concretely as
\begin{align*}
    \left\langle
\mathbb F\mathcal L_{t,\zeta}(q,v_q),w_q
\right\rangle
\left.=\frac{d}{d\epsilon}\right|_{\epsilon=0}
\mathcal L_{t,\zeta}(q,v_q+\epsilon w_q),\;\;\;\;\; q\in Q_t, v_q,w_q\in T_qQ_t.
\end{align*}
This clearly shows that the differentiation occurs only in the vertical velocity direction. Note that from now on we will use the notation $q\in Q_t,v_q\in T_qQ_t$ when we think ``instantaneously'', whereas we will revert back to the $\varphi_t,\dot{\varphi}_t$ notation when we think of the instantaneous state as descending from a covariant field history $\varphi\in \Gamma(Y)$.
\begin{example}\label{ex:kineticlegendre}
Suppose that
\[
\mathcal L_{t,\zeta}(q,v_q)
=
\frac12 g_q(v_q,v_q)-V(q),
\qquad v\in T_qQ_t,
\]
where \(g_q\colon T_qQ_t\times T_qQ_t\to\R\) is a symmetric bilinear form on the admissible velocities.
Then \cite{binz}
\[
\mathbb F\mathcal L_{t,\zeta}(q,v_q)
=
g_q(v_q,\cdot)\in T_q^*Q_t.
\]
\end{example}
One can also define a covariant Legendre transform $\mathbb{F}\mathcal{L}$ from $J^1Y$ to the so-called \textit{multiphase space} $Z$, which is canonically isomorphic to the affine dual $J^1Y^\star$ \cite{gotayMomentumMapsClassical2004}. The image of this covariant Legendre transform is the covariant primary constraint set. The covariant and instantaneous Legendre transforms commute in the sense that the diagram
\begin{equation}\label{eq:legendre-square}
\begin{tikzcd}
(j^1Q)_t
    \arrow[r, "\mathbb{F}\mathfrak{L}"]
    \arrow[d, "\beta_\zeta"']
&
Z_t
    \arrow[d, "R_t"]
\\
TQ_t
    \arrow[r, "\mathbb{F}\mathcal{L}_{t,\zeta}"']
&
T^*Q_t .
\end{tikzcd}
\end{equation}
commutes \cite{gotayMomentumMapsClassical2004a}. The map
$R_t\colon Z_t\to T^*Q_t$ associates a canonical momentum to a multimomentum section in $Z_t$. Thus, taking the covariant Legendre transform and subsequently performing the $3+1$ decomposition yields the same canonical momentum as first decomposing the covariant jet and then taking the instantaneous Legendre transform. 

Henceforth we fix a compatible Lagrangian slicing $(\zeta_M,\zeta_Y)$ and suppress its
dependence from the notation where possible.

\section{From covariant to instantaneous boundary conditions}\label{sec:BCs}

Having clarified the terminology of covariant and instantaneous Lagrangians and Lagrangian densities, we proceed to a novel study of boundaries in this framework. The closest work we are aware of in this direction is the dissertation by Kur \cite{Kur2018MultisymplecticGW}, which focuses on the multisymplectic formalism. We, however, focus on the other three corners of the square \eqref{eq:legendre-square}.

If spacetime $M$ has a boundary $\partial M$, then a covariant boundary condition is defined as \cite{binz}:
\begin{definition}
    A \textit{boundary condition} is a choice of subbundle $B\subset J^1Y|_{\partial M}$, where $Y\to M$ denotes the field bundle. One then considers only those sections $\varphi\in\Gamma(Y)$ for which $j^1\varphi|_{\partial M}\in B$.
\end{definition}
The best-known example would be a \textit{Dirichlet boundary condition}, i.e. a choice of section $b\in \Gamma(Y|_{\partial M})$ and allowing only fields which agree with $b$ over $\partial M$.

With the slicing $\Sigma_t$ of $M$ into Cauchy surfaces we call $\partial\Sigma_t$ the \textit{instantaneous spatial boundary} at time $t$ and we define the \textit{total spatial boundary} as $\bigsqcup_t \partial\Sigma_t$. We assume the slicing to preserve the boundary, i.e. $\zeta_M|_{\partial M}\in T\partial M$, such that its flow by $\tau$ transports $\partial\Sigma_t$ to $\partial\Sigma_{t+\tau}$. The jet decomposition $\beta_\zeta$ then restricts to the boundary:
\[
\beta_{\zeta}^\partial:
(J^1Y)|_{\partial\Sigma_t}
\overset{\sim}{\longrightarrow}
\left(J^1(Y_t)\times_{Y_t}VY_t\right)|_{\partial\Sigma_t}.
\]
Hence a covariant boundary condition \(B\subset J^1Y|_{\partial M}\)
induces the instantaneous boundary condition
\[
B_t^\zeta:=\beta_{\zeta}^\partial(B|_{\partial\Sigma_t}).
\]
Then for a field \(\varphi\):
\[
j^1\varphi|_{\partial\Sigma_t}\in B|_{\partial\Sigma_t}
\quad\Longleftrightarrow\quad
(j^1\varphi_t,\dot\varphi_t)|_{\partial\Sigma_t}\in B_t^\zeta.
\]
We denote the instantaneous boundary restriction by $r_t\colon Q_t=\Gamma(Y_t)\to \Gamma(Y_t|_{\partial\Sigma_t})$, sending $r_t(q)=q|_{\partial\Sigma_t}$. We assume that the relevant spaces of sections are smooth manifolds and that $r_t$ is a smooth surjective submersion.
\begin{lemma}\label{kernellemma}
For $b_t\in \Gamma(Y_t|_{\partial\Sigma_t})$, we define $Q_t^{b_t}=r_t^{-1}(b_t)$. It then follows that
\begin{align*}
    T_qQ_t^{b_t}
    &= \ker(T_qr_t) = \left\{v\in T_qQ_t :
    v|_{\partial\Sigma_t}=0\right\}.
\end{align*}
\end{lemma}
\begin{proof}
    Since $r_t$ is a smooth submersion, $Q_t^{b_t}=r_t^{-1}(b_t)$
is a smooth submanifold and $T_qQ_t^{b_t}=\ker(T_qr_t)$.
Moreover, $T_qr_t(v)=v|_{\partial\Sigma_t}$, which proves the
result.
\end{proof}
We call such a specification $b_t\in \Gamma(Y_t|_{\partial\Sigma_t})$ an \textit{instantaneous spatial Dirichlet boundary condition}. Let $F_\tau^M$ and $F_\tau^Y$ denote the flows of $\zeta_M,\zeta_Y$. Define the transported pullback by 
\begin{align*}
    F_\tau^*\varphi_{t+\tau}
    :=
    (F_\tau^Y)^{-1}
    \circ \varphi_{t+\tau}
    \circ F_\tau^M,\;\;\;\;\; \varphi\in \Gamma(Y).
\end{align*}
We then define the following.
\begin{definition}
A smooth family $b_t\in \Gamma(Y_t|_{\partial\Sigma_t})$ of instantaneous spatial Dirichlet boundary conditions is called \textit{constant} with
respect to the slicing $(\zeta_M,\zeta_Y)$ if $F_\tau^*b_{t+\tau}=b_t$
for every $t$ and $\tau$.
\end{definition}
The following result then formalizes the idea that the velocity of a constant boundary condition should vanish.
\begin{lemma}\label{stupidlemma}
    Let $\varphi\in \Gamma(Y)$ be a field history. Then the family of boundary field values $\varphi_t|_{\partial\Sigma_t}=r_t(\varphi_t)$ is constant if and only if $\dot{\varphi}_t|_{\partial\Sigma_t}=0$ for every $t$.
\end{lemma}

\begin{proof}
By the definition of the flow associated to the slicing,
\begin{align*}
    \frac{d}{d\tau}
    F_\tau^*(\varphi_{t+\tau}|_{\partial\Sigma_{t+\tau}})
    =
    F_\tau^*(\dot{\varphi}_{t+\tau}
    |_{\partial\Sigma_{t+\tau}}).
\end{align*}
Thus a constant family of boundary states (for which the LHS vanishes) has vanishing boundary velocity for all $t$.
Conversely, if the boundary velocity vanishes at every time, then
$F_\tau^*\left(\varphi_{t+\tau}|_{\partial\Sigma_{t+\tau}}\right)$ has zero derivative with respect to $\tau$ and
is therefore equal to $\varphi_t|_{\partial\Sigma_t}$.
\end{proof}

This simple result just formalizes the idea that a family of instantaneous states whose time derivatives are required to be zero are non-dynamical, i.e. when a configuration is specified at any time, the state must be in this configuration at all times. We have simply applied this idea to the specific case in which this restriction of the time derivative is made only on the boundary.

We now proceed to prove a characterization of the instantaneous state space in the presence of a constant Dirichlet boundary condition. 

\begin{proposition}\label{beautifulprop}
    Let $b$ denote a constant spatial Dirichlet boundary condition, whose restriction to $\partial \Sigma_t$ is denoted $b_t$. Denote by $(j^1Q^b)_t$ the collection of holonomic sections agreeing with $b$ on the spatial boundary, restricted to $\Sigma_t$. Then the isomorphism $(j^1Q)_t\cong TQ_t$ of Proposition \ref{prop:crucialisom}, induced by the jet decomposition map $\beta_\zeta$, restrict to an isomorphism $(j^1Q^b)_t\cong TQ^{b_t}_t$.
\end{proposition}

\begin{proof}
We restrict the jet decomposition map $\beta_\zeta$ to elements of
$(j^1Q^b)_t$. Let $j^1\varphi\circ i_t\in(j^1Q^b)_t$. Since
$\varphi$ agrees with $b$ on the spatial boundary, we have $\varphi_t\in Q_t^{b_t}$.
Moreover, since $b$ is constant, Lemma \ref{stupidlemma} implies $\dot{\varphi}_t|_{\partial\Sigma_t}=0$. It then follows from Lemma \ref{kernellemma} that
\begin{align*}
    \dot{\varphi}_t
    \in
    \ker(T_{\varphi_t}r_t)
    =
    T_{\varphi_t}Q_t^{b_t}.
\end{align*}
Thus $\beta_\zeta(j^1\varphi\circ i_t)
    =(\varphi_t,\dot{\varphi}_t)$ lies in $TQ_t^{b_t}$ under the isomorphism of Proposition \ref{prop:crucialisom}.

Conversely, let $(q,v)\in TQ_t^{b_t}$. Since $Q_t^{b_t}=r_t^{-1}(b_t)$ is a
smooth submanifold of $Q_t$, there exists a smooth curve $\gamma:I\longrightarrow Q_t^{b_t}$ such that $\gamma(0)=q$ and $\dot{\gamma}(0)=v$. For $s\in I$, define a field on the
neighboring slice $\Sigma_{t+s}$ by
\begin{align*}
    \varphi_{t+s}
    :=
    F_s^Y\circ\gamma(s)\circ(F_s^M)^{-1}.
\end{align*}
Clearly $\varphi_t=\gamma(0)=q$. These fields define a field history in a neighborhood of
$\Sigma_t$. Since $\gamma(s)|_{\partial\Sigma_t}=b_t$ and $b$ is
constant with respect to the slicing, we have
\begin{align*}
    \varphi_{t+s}|_{\partial\Sigma_{t+s}}
    =
    F_s^Y\circ b_t\circ(F_s^M)^{-1}
    =
    b_{t+s}.
\end{align*}
Thus $\varphi$ agrees with $b$ on the spatial boundary. Moreover $F_s^*\varphi_{t+s}=\gamma(s)$,
and therefore
\begin{align*}
    \dot{\varphi}_t
    =
    \left.\frac{d}{ds}\right|_{s=0}
    F_s^*\varphi_{t+s}
    =
    \dot{\gamma}(0)
    =
    v.
\end{align*}
Consequently, $(q,v)$ is the image of $j^1\varphi\circ i_t$ under the restriction of
$\beta_\zeta$ to $(j^1Q^b)_t$, using the isomorphism of Proposition
\ref{prop:crucialisom}. By the injectivity of $\beta_\zeta$ this restriction is itself an isomorphism.
\end{proof}

Proposition \ref{beautifulprop} demonstrates that, when a constant spatial Dirichlet boundary condition is imposed on the covariant configuration space, the corresponding instantaneous state space is a tangent bundle. However, one may instead encounter an instantaneous boundary condition which restricts only the admissible \textit{velocities},
without selecting a particular boundary value for the instantaneous configurations.
This will occur in Yang--Mills theory, where finiteness of the
instantaneous Lagrangian imposes asymptotic conditions on the
velocities (Section \ref{sec:YM}). We must therefore also understand the instantaneous state space obtained by allowing all boundary
values, while requiring the velocity to vanish on the boundary, a case to which Proposition \ref{beautifulprop} seems inapplicable.

\begin{corollary}\label{funcorollary}
Let $Y\to M$ be a field bundle over a spacetime sliced into Cauchy
surfaces $\Sigma_t$ with boundaries $\partial\Sigma_t$. Denote by
$Y_t$ the restriction of $Y$ over $\Sigma_t$, and let
$Q_t=\Gamma(Y_t)$ be the instantaneous configuration space at time $t$. Let $D_t\subset TQ_t$ denote the instantaneous state space obtained
from histories $\varphi\in\Gamma(Y)$ satisfying $\dot{\varphi}_t|_{\partial\Sigma_t}=0$
for every time $t$. Then
\begin{align}\label{eq:corollary}
    D_t
    =
    \bigsqcup_{b_t\in\Gamma(Y_t|_{\partial\Sigma_t})}
    TQ_t^{b_t},
\end{align}
where $Q_t^{b_t}:=r_t^{-1}(b_t)$
is the space of configurations agreeing with $b_t$ on
$\partial\Sigma_t$.
\end{corollary}

\begin{proof}
Let $\varphi$ be such a field history. By Lemma \ref{stupidlemma}, the boundary
values $b_t:=\varphi_t|_{\partial\Sigma_t}$
define a constant spatial Dirichlet boundary condition. Hence, by
Proposition \ref{beautifulprop}, the instantaneous state obtained from
$\varphi$ at time $t$ lies in $TQ_t^{b_t}$. This proves that every element of $D_t$ lies in one of the summands on
the right-hand side.

Conversely, let $(q,v)\in TQ_t^{b_t}$
for some $b_t\in\Gamma(Y_t|_{\partial\Sigma_t})$. Let $b$ be the
constant spatial Dirichlet boundary condition generated from $b_t$ by
the slicing. Proposition \ref{beautifulprop} says that $(q,v)$ is
obtained, under the jet decomposition map, from a holonomic section
satisfying this constant boundary condition. By Lemma
\ref{stupidlemma}, such a history has vanishing velocity on
the spatial boundary at every time. Therefore $(q,v)\in D_t$.

Finally, the union is disjoint because the spaces $Q_t^{b_t}$ are the
fibers of the boundary restriction map $r_t$. Equivalently, by Lemma
\ref{kernellemma}: $TQ_t^{b_t}
    =
    \ker(Tr_t)|_{Q_t^{b_t}},$ so the RHS of Eq.~\eqref{eq:corollary} is precisely the decomposition of the vanishing-boundary-velocity states into their fixed-boundary-value sectors.
\end{proof}

Corollary \ref{funcorollary} shows that a field history whose velocity vanishes on the spatial boundary at every time necessarily has some constant spatial boundary value. But since this value is not specified, the resulting instantaneous state space is not \textit{one} fixed tangent bundle $TQ_t^{b_t}$, but the disjoint union of all such fixed-boundary sectors. Thus the fibers of $r_t$ form
dynamically preserved boundary sectors, and the corresponding leaves are the tangent bundles $TQ_t^{b_t}$ of the fixed-Dirichlet configuration spaces. From now on we suppress the time subscript on $Q_t$, $Y_t$, $\Sigma_t$, $r_t$ and $D_t$ where possible.

\section{Instantaneous Legendre transform with boundaries}\label{sec:legendre}

The instantaneous Legendre transform (Definition \ref{def:instlegendre}) is defined on the tangent bundle to the instantaneous configuration space. If a spatial Dirichlet boundary condition is present, then by Proposition \ref{beautifulprop} the instantaneous state space is still a tangent bundle and one can straightforwardly perform the Legendre transform with boundary condition.

But this is less straightforward in the case covered by Corollary \ref{funcorollary}, as a priori there seem to be three possibilities for defining the Legendre transform. If the Lagrangian $\mathcal{L}$ is defined on the full space $TQ$, then one might just restrict its Legendre transform $\mathbb{F}\mathcal{L}$ to the sum $D_t$ of disjoint sectors. But the Lagrangian might be ill-defined on those states in the full tangent bundle which do not satisfy the boundary condition (the vanishing of velocity on the boundary in the case of the corollary), for instance in the case in which $\partial\Sigma$ is the conformal boundary of a non-compact Cauchy surface (we will see this happen for Yang-Mills theory in Section \ref{sec:YM}). In that case there are two other obvious things to consider: performing the Legendre transform on each separate sector tangent bundle, or recalculate it on the full space of all sectors $D$, using a definition of fiber derivative that works also when the domain is not strictly a tangent bundle. The question is whether these various options agree.

Let $r\colon Q\to Q^\partial$ be the boundary restriction map, where $Q^\partial\subset \Gamma(Y_t|_{\partial\Sigma_t})$ denotes the space of allowed boundary data (earlier we assumed this to equal all of $\Gamma(Y_t|_{\partial\Sigma_t})$, but one could in principle impose further conditions), and assume that \(r\) is a smooth split
surjective submersion. The admissible velocity bundle $D=\text{ker}(Tr)=\bigsqcup_{b\in Q^\partial}TQ^b$ is not a tangent bundle, but it is still a vector bundle with typical fibers $D_q=T_qQ^b$ for $q\in Q^b$. We can therefore still define its dual bundle $D^*\to Q$, with fiber $D_q^*:=(D_q)^*\cong T_q^* Q^b$. So $D^*|_{Q^b} \cong T^*Q^b$. Clearly $D^*$ retains the sector structure of $D$. Intuitively, this just means that there are no boundary conjugate momenta, since there were no boundary velocities in $D$.

\begin{remark}\label{boundaryremark}
    It is crucial to note that $D^*=\bigsqcup_{b\in Q^\partial} T^*Q^b$ is not necessarily a subspace of $T^*Q$. It might be that certain elements of $T^*Q^b$ are well-defined precisely because they need to act only on tangent vectors $v\in T_q Q^b$ to configurations that equal $b$ on $\partial\Sigma$, whereas co-vectors in $T_q^*Q$ would need to be well-defined on the larger set of vectors $T_qQ$. We will see an example of this in Section \ref{sec:YM}.
\end{remark}

We now define the Legendre transform on all of $D$ by simply using the general definition of the fiber derivative on vector bundles.

\begin{definition}
Let \(\mathcal{L}_D\colon D\to\mathbb R\) be an instantaneous Lagrangian. Its Legendre
transform is the fiber derivative $\mathbb{F}\mathcal{L}_D\colon D\to D^*$
defined by
\[
\left\langle \mathbb F_D \mathcal L_D(v_q),w_q\right\rangle
=
\left.\frac{d}{d\epsilon}\right|_{\epsilon=0}
\mathcal L_D(v_q+\epsilon w_q),
\qquad
v_q,w_q\in D_q.
\]
\end{definition}

This works because we sum only vectors tangent to the same sector $Q^b$. We now address the question as to whether this Legendre transform agrees with the other possible definitions.

\begin{proposition}\label{prop:legendresectors}
Suppose one has a Lagrangian $\mathcal{L}_D\colon D\to \R$. For every \(b\in Q^\partial\), let $\mathcal L_b:=\mathcal L_D|_{TQ^b}$. Then the Legendre transform on \(D\) restricts to the ordinary Legendre
transform on each sector:
\[
\left.\mathbb F\mathcal L_D\right|_{TQ^b}
=
\mathbb F\mathcal L_b.
\]
Suppose, moreover, that there exists a Lagrangian $\mathcal L\colon TQ\longrightarrow\mathbb R$ such that \(\mathcal L_D=\mathcal L|_D\). Then
\begin{align*}
\mathbb F\mathcal L_D(v_q)
=
\left.\mathbb F\mathcal L(v_q)\right|_{D_q},\;\;\;\;\; v_q\in D_q=T_qQ^b.
\end{align*}
Consequently $\mathbb F\mathcal L_b(v_q)
=
\left.\mathbb F\mathcal L(v_q)\right|_{T_qQ^b}$ for $q\in Q^b$. Thus the three possible constructions of the Legendre transform agree
wherever they are simultaneously defined.
\end{proposition}
\begin{proof}
Let \(v_q,w_q\in D_q=T_qQ^b\). Then
\begin{align*}
\left\langle\mathbb F\mathcal L_D(v_q),w_q\right\rangle
&=
\left.\frac{d}{d\epsilon}\right|_{\epsilon=0}
\mathcal L_D(v_q+\epsilon w_q)=
\left.\frac{d}{d\epsilon}\right|_{\epsilon=0}
\mathcal L_b(v_q+\epsilon w_q)=
\left\langle\mathbb F\mathcal L_b(v_q),w_q\right\rangle.
\end{align*}
This proves the first statement. If
\(\mathcal L_D=\mathcal L|_D\), the same calculation gives $\left\langle\mathbb F\mathcal L_D(v_q),w_q\right\rangle
=
\left\langle\mathbb F\mathcal L(v_q),w_q\right\rangle$ for every \(w_q\in D_q\), proving the rest.
\end{proof}
Note that the covector $\mathbb{F}\mathcal{L}(v_q)$ must be restricted to $D_q=T_q Q^b$ because a priori it has the larger domain $T_q Q$.

This Proposition is not very deep, but its interest actually comes from the case in which there is \textit{no} Lagrangian defined on all of $TQ$, as will be the case in Section \ref{sec:YM}. Then one can perform the Legendre transform only sectorwise, but the Proposition ensures that it would have equalled the restriction of the Legendre transform of the general Lagrangian if that had existed.

With a notion of Legendre transform in the presence of boundaries at hand, we might ask whether there is something akin to the GiMmsy square \eqref{eq:legendre-square}. Indeed, we fully expect the commuting of the covariant and instantaneous Legendre transforms to hold also in the presence of boundaries, when the domains are restricted appropriately to jets satisfying the boundary condition. The proof of \cite[Proposition 6.3]{gotayMomentumMapsClassical2004a} still holds when the map $R_t\colon Z_t\to T^*Q_t$, defined in Eq. (5D.1) of \cite{gotayMomentumMapsClassical2004a}, is restricted to the image of the admissible jets under $\mathbb{F}\mathfrak{L}$ and to admissible variations in $D_q$.

We end this section by considering the structure of the image of the Legendre transform of $D$, i.e. the phase space in the presence of boundaries. Just as for the ordinary case without boundaries we define
\begin{definition}
The primary phase space associated with \(\mathcal L_D\) is $\mathcal{P}:=\operatorname{Im}(\mathbb F\mathcal L_D)\subset D^*$.
For each \(b\in Q^\partial\), let $\mathcal{P}_b:=\operatorname{Im}(\mathbb F\mathcal L_b)
\subset T^*Q^b$.
\end{definition}
Clearly, by Proposition \ref{prop:legendresectors}, the primary phase space decomposes: $\mathcal P=\bigsqcup_{b\in Q^\partial}\mathcal{P}_b$.
If \(\mathcal L_D\) is hyperregular, then $P=D^*
=
\bigsqcup_{b\in Q^\partial}T^*Q^b$.
In the almost-regular case, we assume each \(\mathcal{P}_b\) is a smooth closed submanifold of \(T^*Q^b\) and each sectorwise Legendre transform is a submersion with connected fibres. Then clearly \(\mathcal{P}_b\) carries the usual primary presymplectic form obtained by restricting the canonical weak symplectic form on \(T^*Q^b\). As usual in infinite-dimensional field theory, the symplectic form is only weakly
nondegenerate and need not define
an isomorphism \(T(T^*Q^b)\cong T^*(T^*Q^b)\). Crucially, however, the total primary phase space $P$ does not carry a global symplectic structure. Intuitively, there are no momenta for the boundary configurations. This is related to what physicists call \textit{edge modes} \cite{donnellyLocalSubsystemsGauge2016a,rielloEdgeModesEdge2021}.

\begin{remark}
Though the total phase space $D^*$ does not carry a global symplectic structure, it can be equipped with a Poisson structure that combines the canonical Poisson structures on the cotangent bundles, i.e. by defining
\[
\left.\{F,G\}_{D^*}\right|_{T^*Q^b}
=
\left\{
F|_{T^*Q^b},
G|_{T^*Q^b}
\right\}_{T^*Q^b},
\]
for admissible functions $F,G\colon D^*\to\R$. Its symplectic leaves are the spaces \(T^*Q^b\), or their connected components.
\end{remark}

\section{Application to Yang-Mills theory}\label{sec:YM}

So far we have not been concerned at all with field theories possessing a localizable symmetry group, although our whole motivation for this article is the interaction between such symmetry groups and boundary conditions, as outlined in Section \ref{intro}. We now proceed to an application of the results developed above to Yang-Mills theory on a Cauchy surface in Minkowski spacetime, which includes spatial asymptotic boundary conditions that can be translated to the framework developed in this paper by means of a conformal compactification. The case covered by Proposition \ref{beautifulprop} was studied at length in \cite{Borsboom:2025agn}, where it was assumed that the instantaneous state space of Yang-Mills theory in temporal gauge is a tangent bundle. We now extend this to the full generality of Corollary \ref{funcorollary}.

We recall the basics from \cite{Borsboom:2025agn}. We consider a principal $G$-bundle $P\to\Sigma\cong \R^3$, with compact structure group with Lie algebra $\mathfrak{g}$, and gauge group $\mathcal{G}=\text{Aut}(P)$. The fact that we consider configurations over $\Sigma$ and not $M$ refers to our working in temporal gauge, in which $A_0=0$. After a conformal compactification of $M$ we instead work on compact $\hat{\Sigma}$ with boundary $\partial\hat{\Sigma}\cong S^2$, and denote all objects on this compactified space with a hat. The total configuration space before imposing boundary conditions would be $Q=\text{Conn}(P)\cong \Omega^1(\Sigma,\text{Ad}(P))$ (as affine spaces), or $\hat{Q}=\text{Conn}(\hat{P})$. We denote velocities by $\alpha_A\in T_AQ$. The instantaneous Lagrangian is
\[
\mathcal L(A,\alpha_A)
=
\frac12\|\alpha_A\|^2
-
\frac12\|F(A)\|^2,
\]
where the norm on $\mathfrak{g}$-valued forms is defined by $\norm{\omega}^2=\int_\Sigma \text{Tr }\omega\wedge *\omega$. Clearly the norm on velocities $\alpha_A$ and curvatures $F(A)$ is finite only if these fall off sufficiently quickly towards infinity. The required fall-off rates were studied in \cite{Borsboom:2025agn}. For the norm $\norm{\alpha_A}$ to exist we assume it to be the pullback of some $\hat{\alpha}\in \Omega^1(\hat{\Sigma};\text{Ad}(\hat{P}))$ and require that $\hat{\alpha}|_{\partial\hat{\Sigma}}=0$, which guarantees square-integrability \cite[Proposition 3.2]{Borsboom:2025agn}. For $\norm{F(A)}$ to exist we only need to assume that $A$ equals the pullback of some $\hat{A}\in\text{Conn}(\hat{P})$ that extends smoothly to $\partial\hat{\Sigma}$ \cite[Proposition 3.3]{Borsboom:2025agn}.

From now on we work solely on the compactified space $\hat{\Sigma}$ and drop all hats. The boundary condition is $\alpha_A|_{\partial\Sigma}=0$, so by Corollary \ref{funcorollary} the instantaneous Yang--Mills state
space is
\[
D_{\mathrm{YM}}
=
\left\{
(A,\alpha_A)\in TQ:
\alpha_A|_{\partial\Sigma}=0
\right\}
=
\bigsqcup_{b\in Q^\partial}TQ^b,
\]
where $Q^\partial=\text{Conn}(P|_{\partial\Sigma})$. The Legendre transform is of the type of Example \ref{ex:kineticlegendre} with metric $\langle\omega,\upsilon\rangle_{L^2}=\int_\Sigma\text{Tr } \omega\wedge*\upsilon$:
\[
\left\langle
\mathbb F\mathcal L_D(A,\alpha_A),\omega
\right\rangle
=\langle \alpha_A,\omega\rangle_{L^2}=
\int_\Sigma\operatorname{Tr}(\alpha_A\wedge *\omega),\;\;\;\;\;\omega\in T_AQ^b.
\]
Consequently, the Yang--Mills primary phase space decomposes as $\mathcal{P}_{\mathrm{YM}}
=
\bigsqcup_{b\in Q^\partial}\mathcal{P}_b$ with $
\mathcal{P}_b\subset T^*Q^b$. Here \(\mathcal{P}_b\) may be a proper smooth subspace of \(T^*Q^b\), depending on the function spaces chosen to represent smooth cotangent vectors. This is an analytical restriction and should be distinguished from the Gauss constraint, which defines the secondary constraint surface
inside \(\mathcal{P}_b\). We write $\mathbb F\mathcal L_D(A,\alpha_A)=(A,E)$, where $E=*\alpha_A\in \Omega^2(\Sigma,\mathfrak{g})$ is the electric field. This is the example foreshadowed in Remark \ref{boundaryremark}: $E$ acts only on variations vanishing at the boundary, even though the boundary electric flux can be nonzero (see Section 4.2 in \cite{Borsboom:2025agn}).

Having established the structure of the instantaneous state and phase spaces, we turn to the action of the gauge group $\mathcal{G}$ in order to derive the group of physical gauge transformations. It acts in the usual way: 
\[
A^g=g^{-1}Ag+g^{-1}dg,
\qquad
\alpha_A^g=g^{-1}\alpha_Ag,
\qquad
E^g=g^{-1}Eg.
\]
\begin{proposition}\label{prop:gaugesectors}
The full gauge group \(\mathcal G\) acts on
\(D_{\mathrm{YM}}\) and \(\mathcal{P}_{\mathrm{YM}}\). Since $r(A^g)=r(A)^{g_\partial}$ (where $r$ is the boundary restriction map) we have
\[
g\cdot TQ^b=TQ^{b^{g_\partial}},
\qquad
g\cdot \mathcal{P}_b=\mathcal P_{b^{g_\partial}}.
\]
The full gauge group therefore acts by Poisson automorphisms of \(\mathcal{P}_{\mathrm{YM}}\), but it need not preserve its symplectic leaves.
\end{proposition}

\begin{proof}
A smooth gauge transformation maps a smooth connection to a smooth
connection. Moreover $\alpha_A|_{\partial\Sigma}=0$ clearly implies $g^{-1}\alpha_Ag|_{\partial\Sigma}=0$.
The map between \(\mathcal{P}_b\) and \(\mathcal P_{b^{g_\partial}}\) is the cotangent lift of the corresponding map between configuration spaces, and therefore preserves their canonical Poisson structures.
\end{proof}
We denote the subgroup preserving the sector labelled by \(b\) by $\mathcal G_b
:=
\{g\in\mathcal G:b^{g_\partial}=b\}.$
Its Lie algebra is $\mathfrak g_b
=
\{\xi\in\operatorname{Lie}(\mathcal G):
D_b\xi_\partial=0\}$. 
The infinitesimal gauge transformation generated by \(\xi\) is tangent
to \(\mathcal{P}_b\) if and only if $D_b\xi_\partial=0$.
If \(D_b\xi_\partial\neq0\), it moves between symplectic leaves and
cannot be generated by a Hamiltonian function for the Poisson
structure on \(\mathcal{P}_{\mathrm{YM}}\), as follows from the fact that Hamiltonian vector fields on a Poisson manifold are tangent to its
symplectic leaves.

The question now becomes how to understand the Gauss law constraint $D_AE=0$ in this framework. Even though in temporal gauge there is no $A_0$ component whose velocity is absent in the Lagrangian, we must of course still impose the Gauss law in order to preserve this gauge. Normally the Gauss law defines the secondary constraint in one particular sector, but this can straightforwardly be extended to all sectors simultaneously. In \cite{Borsboom:2025agn} it was shown that the Gauss law is the momentum map for the subgroup $G^\infty_0\subset\mathcal{G}$ of gauge transformations which vanish on the boundary and belong to the connected component of the identity transformation. The $\infty$ superscript there refers to the conformal boundary ``at infinity''. We retain this superscript here. We then have the following result.

\begin{proposition}\label{prop:allsectorGauss}
The action of \(\mathcal G_0^\infty\) is Hamiltonian on every leaf
\(\mathcal{P}_b\), with momentum map given by the Gauss law: $J(A,E)=D_AE$.
Its zero locus is
\[
\mathcal C
=
\bigsqcup_{b\in Q^\partial}\mathcal C_b,
\qquad
\mathcal C_b
=
\{(A,E)\in \mathcal{P}_b:D_AE=0\}.
\]
Furthermore $g\cdot\mathcal C_b=\mathcal C_{b^{g_\partial}}$
for every \(g\in\mathcal G\).
\end{proposition}

\begin{proof}
For \(\xi|_{\partial\Sigma}=0\) (an element of $\mathfrak G^\infty:=\text{Lie}(\mathcal{G}^\infty_0)$, integration by parts of the smeared Gauss constraint $\int_\Sigma \text{Tr }\xi D_AE$ produces no
boundary term, so the ordinary Yang--Mills momentum-map calculation from Section 4.1 of \cite{Borsboom:2025agn}
applies on every leaf. The last claim follows from the gauge covariance of the constraint: $D_{A^g}E^g=g^{-1}(D_AE)g$.\footnote{This is straightforwardly verified:
\begin{align*}
D_{A^g}E^g&=dE^g+[A^g,E^g]=d(g^{-1}Eg)+[g^{-1}Ag+g^{-1}dg,g^{-1}Eg]
\\
&=d(g^{-1})Eg+g^{-1}(dE)g+g^{-1}Edg+g^{-1}[A,E]g+g^{-1}(dg)g^{-1}Eg-g^{-1}Edg
\\
&=g^{-1}(D_AE)g+d(g^{-1})Eg+g^{-1}(dg)g^{-1}Eg
\\
&=g^{-1}(D_AE)g-g^{-1}(dg)g^{-1}Eg+g^{-1}(dg)g^{-1}Eg=g^{-1}(D_AE)g.
\end{align*}}
\end{proof}

Thus the action of the constraint subgroupgroup on each leaf is Hamiltonian, but the sector-moving gauge transformations cannot be so. Recall the definition of momentum map in Poisson geometry \cite{FernandesOrtegaRatiu2009}: if a Lie group \(H\) acts by Poisson diffeomorphisms on a Poisson
manifold \((P,\Pi)\), then the action is called Hamiltonian if there exists
an equivariant momentum map $J:P\longrightarrow\mathfrak h^*$
such that
\[
X_\xi
=
\Pi^\sharp dJ^\xi,
\qquad
J^\xi:=\langle J,\xi\rangle,
\]
for every \(\xi\in\mathfrak h\). But since the image of \(\Pi^\sharp\) is precisely the tangent distribution of
the symplectic foliation, any Hamiltonian vector field $X_f=\Pi^\sharp(df)$ is tangent to the symplectic leaves. Consequently, an action
which moves points between different symplectic leaves cannot possess a momentum map. Thus the sector-moving gauge transformations carry no momentum. 

Crucially, however, the group of sector-preserving gauge transformations, denoted $\mathcal{G}^b=\{g\in \mathcal{G}:g^\partial\cdot b=b\}$, is larger than $\mathcal{G}^\infty_0$. The action of this larger group is still Hamiltonian.

\begin{proposition}\label{prop:sectorboundarymomentum}
The action of \(\mathcal G_b\) on \(\mathcal{P}_b\) has momentum map $\left\langle J_b(A,E),\xi\right\rangle
=
-\int_\Sigma\operatorname{Tr}(E\wedge D_A\xi)$ for $
\xi\in\mathfrak g_b$.
Integration by parts gives
\[
\left\langle J_b(A,E),\xi\right\rangle
=
\int_\Sigma\operatorname{Tr}((D_AE)\xi)
-
\int_{\partial\Sigma}\operatorname{Tr}(E\,\xi_\partial).
\]
Therefore, on the constraint surface \(\mathcal C_b\) we have boundary momentum $J_b^\xi|_{\mathcal C_b}
=
-\int_{\partial\Sigma}\operatorname{Tr}(E\,\xi_\partial)$.
\end{proposition}
\begin{proof}
    See proof of \cite[Proposition 4.2]{Borsboom:2025agn}.
\end{proof}
This just expresses the well-known fact that the allowed boundary gauge transformations carry a momentum equal to the electric flux through the boundary, see also \cite{rielloHamiltonianGaugeTheory2024}. Thus we see that there are three types of gauge transformations: 
\begin{itemize}
    \item Those satisfying $\xi|_{\partial\Sigma}=0$, which are generated by the Gauss constraint and ``trivial'', ``redundant'' or ``unphysical'';
    \item Those which do not vanish on the boundary, $\xi|_{\partial\Sigma}\neq 0$, but are covariantly constant: $D_b\xi|_{\partial\Sigma}=0$. The action of these is Hamiltonian with boundary electric flux momentum;
    \item The sector moving ones, $D_b\xi|_{\partial\Sigma}\neq 0$. These do not carry momentum on $\mathcal{P}_\text{YM}$.
\end{itemize}

We now propose that the usual definition of the physical asymptotic/boundary symmetry group as the quotient of the boundary-preserving transformations by the ``trivial'' transformations should be replaced by:
\begin{definition}\label{def:physicalgaugegroup}
For a gauge theory on a manifold with boundary, the physical gauge subgroup consists of the allowed Hamiltonian
gauge symmetries modulo the trivial gauge symmetries:
\[
\mathcal G_{\mathrm{phys}}
:=
\mathcal{G}_{\mathrm{Ham}}/
     \mathcal G_{\mathrm{triv}},
\]
\end{definition}
We propose this alternative definition because we believe that gauge transformations which carry no corresponding momentum cannot be thought of as physical in any meaningful sense. Boundary transformations can be physical because when two subsystems are compared, one can notice a difference in the momenta of the subsystems:  the momentum is the
quantity through which the transformation of a subsystem can be distinguished by its
interaction with the environment. If there is no momentum then this is impossible.

Before we say more about the full $\mathcal{G}_{\mathrm{Ham}}/
     \mathcal G_{\mathrm{triv}}$ on all of $\mathcal{P}_\text{YM}$, we note the following for the restriction to one sector. The precise form of \(\mathcal G_b/\mathcal G^\infty_0\) depends on the
stabilizer of the boundary connection \(b\). This stabilizer is controlled by
the holonomy of \(b\): for a connection, its gauge stabilizer is the centralizer
of its holonomy group; see \cite{Fleischhack:2000am}. In the case of a trivial holonomy of the boundary configuration $b$, we thus see that the \textit{small} (i.e. connected to the identity) physical gauge transformations in the $b$-sector are just $G$, which reproduces the result of \cite{Borsboom:2025agn}. In the general case, however, global gauge (constant) transformations can have a sector-moving action and be non-Hamiltonian. Define $\mathfrak G_{\mathrm{Ham}}^{\mathrm{tot}}
:=
\left\{
\xi\in\operatorname{Lie}(\mathcal G):
X_\xi\text{ is Hamiltonian on }\mathcal{P}_{\mathrm{YM}}
\right\}.$
Then:
\begin{proposition}\label{prop:totalhamiltoniangroup}
In a boundary trivialization, $\mathfrak G_{\mathrm{Ham}}^{\mathrm{tot}}
/
\mathfrak G^\infty
\cong
\mathfrak z(\mathfrak g)$, where \(\mathfrak z(\mathfrak g)\) is the center of \(\mathfrak g\).
Consequently,
\[
(\mathcal G_{\mathrm{Ham}}^{\mathrm{tot}})_0
/
\mathcal G_0^\infty
\cong
Z(G)_0.
\]
\end{proposition}

\begin{proof}
A globally Hamiltonian infinitesimal transformation must be tangent to
every leaf. Hence $D_b\xi_\partial=0$ for every boundary connection \(b\in Q^\partial\). Taking \(b=0\) (in the trivialization) gives
\(d\xi_\partial=0\), so \(\xi_\partial\) is constant. Allowing arbitrary
\(b\) then gives $[b,\xi_\partial]=0$
for every \(b\). Therefore
\(\xi_\partial\in\mathfrak z(\mathfrak g)\). The converse follows
because a constant central boundary parameter preserves every
boundary connection.
\end{proof}
For semisimple structure group \(G\), this leaves no nontrivial connected physical gauge group on the total all-sector phase space. For $G=U(1)$, however, we obtain the familiar result that $\mathcal{G}_\text{phys}\cong U(1)$ \textit{across all sectors}. Thus, this is a genuine generalization of \cite{Borsboom:2025agn}, and we see that the global gauge group is again singled out as the physical gauge group of electromagnetism because it is sector-preserving but does not vanish on the boundary. This agrees with \cite{gomesQuasilocalDegreesFreedom2021,gomesHolismEmpiricalSignificance2021}. It should be noted that whether this quotient acts nontrivially as a
physical symmetry does depend on the presence of charged matter or other degrees
of freedom with respect to which the corresponding charge can be detected.

For non-Abelian Yang-Mills theory, if we work in a boundary trivialization and allow only $\text{Ad}(G)$-invariant boundary configurations $b$, then we also find $\mathcal{G}_\text{phys}\cong G$, even on the space of all $\text{Ad}(G)$-invariant boundary sectors. This is because the proof of Proposition \ref{prop:totalhamiltoniangroup} then no longer applies because the condition $[b,\xi_\partial]=0$ becomes automatic, and does not hold for arbitrary $b$. Thus the constant $\xi_\partial$ need not lie in the center of $\mathfrak g$. This agrees with \cite{Borsboom:2025agn}, where the physical gauge group for one $\text{Ad}$-invariant sector was shown to be $G$, though admittedly the restriction to $\text{Ad}$-invariant configurations is quite stringent. 

We end this Section with a comparison of our results to \cite{rielloHamiltonianGaugeTheory2024}. The first important point to note is that their sectors are not the same as ours: their so-called \textit{flux superselection sectors} are the symplectic leaves of the \textit{reduced} phase space when the total phase space $T^\vee\text{Conn}(P)$ is quotiented by the constraint group (our $\mathcal{G}^\infty_0$). Their starting phase space is large enough for the full boundary gauge action to admit a local momentum map, whose boundary part is the flux map. In our phase space the boundary-configuration directions have been removed, so transformations changing \(b\) become leaf-moving Poisson
automorphisms and no longer have a classical momentum map on \(\mathcal{P}_{\mathrm{YM}}\). If one fixes \(b\) in our sense, then the Riello--Schiavina flux analysis has a restricted analogue, with the full boundary gauge group replaced by
the stabilizer \(\mathcal G_b/\mathcal G^\infty_0\).

\begin{remark}
In this Section we worked in temporal gauge. Without temporal gauge, finite energy constrains $u=\dot A-D_AA_0$
rather than \(\dot A\). The condition \(u|_{\partial\Sigma}=0\)
then implies $\dot b=D_b(A_0|_{\partial\Sigma})$
for the boundary configuration $b$. Thus the boundary connection may evolve within its boundary gauge
orbit. A fully gauge-covariant treatment would therefore not use the vector bundle
\(D=\ker(Tr)\to Q\). For fixed \(A_0|_{\partial\Sigma}\), the admissible
velocities at \(A\) form the affine space
\[
\{\alpha_A\in T_AQ:T_A r(\alpha_A)=D_{r(A)} A_0|_{\partial\Sigma}\},
\]
modeled on \(\ker(T_A r)\). Thus the Legendre transform would have to be
reformulated for affine admissible-velocity spaces rather than for the vector
bundle used in Sections \ref{sec:BCs} and \ref{sec:legendre}.
\end{remark}

\section{Conclusion}

We have shown that instantaneous boundary conditions can force the admissible state space to be not a single tangent bundle, but a disjoint union of tangent bundles over fixed boundary sectors. The corresponding Legendre transform is then naturally defined on the admissible velocity bundle \(D=\ker(Tr)\), with values in \(D^*\). This yields a sectored Poisson phase space whose symplectic leaves are the fixed-boundary phase spaces.

In Yang--Mills theory this structure explains why transformations changing the boundary connection act between Poisson leaves rather than as Hamiltonian transformations on a single leaf. This suggests that physical boundary gauge transformations should be understood as Hamiltonian gauge transformations modulo trivial ones. In temporal gauge, the Gauss constraint acts sectorwise, while non-trivial sector-preserving boundary transformations carry the usual electric-flux momentum.

Several questions remain open. One should investigate whether analogous Legendre transforms can be defined for more general boundary conditions, including genuinely affine or gauge-covariant ones. The application to gravity \cite{Regge:1974zd}, and especially to null boundaries \cite{Ashtekar:1978zz,Ashtekar:1981bq}, is an important next step. The rigorous relation to the covariant phase space formalism \cite{Crnkovic:1986ex,Lee:1990nz} must also be understood. Finally, edge modes suggest a complementary route: by enlarging the phase space with boundary-conjugate variables \cite{donnellyLocalSubsystemsGauge2016a,rielloEdgeModesEdge2021}, sector-moving transformations may again become Hamiltonian.

\subsection*{Acknowledgments}

The author wants to thank Hessel Posthuma for useful discussions, Klaas Landsman for feedback and Manus Visser for advice. This work is supported by the Spinoza Grant of the Dutch Science Organization (NWO) awarded to N.P. (Klaas) Landsman.

\bibliographystyle{ieeetr} 
\bibliography{bibliography.bib}

\end{document}